%% file: fullversion.tex
\documentclass [11pt] {article}

 \usepackage{fullpage}

\usepackage{graphics}
\usepackage[dvips]{epsfig}

\usepackage{amsmath,amsthm}
\usepackage{amssymb}
\usepackage{amsfonts}
\usepackage{graphicx}
\usepackage{xspace}

\usepackage{algorithm,algorithmicx}

  \newcommand{\algindent}{\hspace{\algorithmicindent}}

\usepackage{comment}
\usepackage{algpseudocode}
\usepackage{xcolor}

\usepackage{cite}

\usepackage{array}
\usepackage{color}

\usepackage[english]{babel}

\usepackage{multirow}
\usepackage{rotating}

\newtheorem{theorem}{Theorem}[section]
\newtheorem{lemma}[theorem]{Lemma}
\newtheorem{corollary}[theorem]{Corollary}

\newtheorem{obs}[theorem]{Observation}

\newtheorem{fact}[theorem]{Fact}

\newcommand{\dom}[1]{\ensuremath{\textrm{DOM}_{#1}}\xspace}
\newcommand{\new}[1]{\ensuremath{\textrm{NEW}_{#1}}\xspace}

\newcommand{\laststage}{\ensuremath{\ell}\xspace}
\newcommand{\inform}[1]{\ensuremath{\textrm{INF}_{#1}}\xspace}
\newcommand{\uninf}[1]{\ensuremath{\textrm{UNINF}_{#1}}\xspace}
\newcommand{\frontier}[1]{\ensuremath{\textrm{FRONTIER}_{#1}}\xspace}

\newcommand{\cB}{{\cal B}}

\newcommand{\cL}{{\cal L}}

\begin{document}
\def\thefootnote{\fnsymbol{footnote}}

\title{{\bf Constant-Length Labeling Schemes for \\
Deterministic Radio Broadcast}}

\author{Faith Ellen\footnotemark[4]
\and Barun Gorain\footnotemark[1]
\and Avery Miller\footnotemark[2]
\and Andrzej Pelc\footnotemark[3]
}

\footnotetext[4]{Department of Computer Science, University of Toronto, {\tt faith@cs.toronto.edu}. Partially supported by NSERC Discovery Grant RGPIN--2015--05080.}
\footnotetext[1]{Indian Institute of Information Technology Vadodara, {\tt barun@iitbhilai.ac.in}.}
\footnotetext[2]{Department of Computer Science, University of Manitoba, \texttt{avery.miller@umanitoba.ca}. Partially supported by NSERC Discovery Grant RGPIN--2017--05936.}
\footnotetext[3]{
 D\'epartement d'informatique, Universit\'e du Qu\'ebec en Outaouais, {\tt pelc@uqo.ca}. Partially supported by NSERC Discovery Grant RGPIN--2013--08136
and by the Research Chair in Distributed Computing at the
Universit\'e du Qu\'ebec en Outaouais.}

\maketitle

\thispagestyle{empty}

\begin{abstract}
Broadcast is one of the fundamental network communication primitives. One node of a network, called the {\em source}, has a message that has to be learned by all other nodes. We consider broadcast in radio networks, modeled as simple undirected connected graphs with a distinguished source. Nodes communicate in synchronous rounds. In each round, a node can either transmit a message to all its neighbours, or stay silent and listen. At the receiving end, a node $v$ hears a message from a neighbour $w$ in a given round if $v$ listens in this round and if $w$ is its only neighbour that transmits in this round. If more than one neighbour of a node $v$ transmits in a given round, we say that a {\em collision} occurs at $v$. We do not assume collision detection: in case of a collision, node $v$ does not hear anything (except the background noise that it also hears when no neighbour transmits).

We are interested in the feasibility of deterministic broadcast in radio networks. If nodes of the network do not have any labels, deterministic broadcast is impossible even in the four-cycle. On the other hand, if all nodes have distinct labels, then broadcast can be carried out, e.g., in a round-robin fashion, and hence $O(\log n)$-bit labels are sufficient for this task in $n$-node networks. In fact, $O(\log \Delta)$-bit labels, where $\Delta$ is the maximum degree, are enough to broadcast successfully. Hence, it is natural to ask if very short labels are sufficient for broadcast. Our main result is a positive answer to this question. We show that every radio network can be labeled using 2 bits in such a way that broadcast can be accomplished by some universal deterministic algorithm that does not know the network topology nor any bound on its size. Moreover, at the expense of an extra bit in the labels, we can get the following additional strong property of our algorithm: there exists a common round in which all nodes know that broadcast has been completed. {Finally, we show that 3-bit labels are also sufficient to solve both versions of broadcast in the case where the labeling scheme does not know which node is the source.}

\vspace*{0.5cm}

\noindent
{\bf keywords:} broadcast, radio network, labeling scheme, feasibility

\vspace*{0.5cm}
\end{abstract}

\pagebreak

\input{labeling-full}


\bibliographystyle{plain}
\bibliography{labeling}

\end{document}

%% file: labeling-full.tex
\section{Introduction}

\subsection{The model and the problem}

Broadcast is one of the fundamental and most extensively studied network communication primitives. 
One node of a network, called the {\em source}, has a message that
has to be learned by all other nodes. We consider broadcast in radio networks, modeled as simple undirected connected graphs with a distinguished source. In the sequel, we use the word {\em graph} in this sense, and we consider the notions of {\em network} and {\em graph} as synonyms.
Nodes communicate in synchronous rounds. Throughout the paper, round numbers refer to the local time at the source, which can differ from the local time at other nodes. In each round, a node can either transmit a message to all its neighbours, or stay silent and listen. At the receiving end, a node $v$ hears a message from a neighbour $w$ in a given round if $v$ listens in this round and if $w$ is its only neighbour that transmits in this round.
If more than one neighbour of a node $v$ transmits in a given round, we say that a {\em collision} occurs at $v$.
We do not assume collision detection: in case of a collision, node $v$ does not hear anything
(except the background noise that it also hears when no neighbour transmits). If collision detection is available, broadcast is trivially feasible, even in anonymous networks: consecutive bits
of the source message can be transmitted by a sequence of silent and noisy rounds, cf.  \cite{CGGPR}, using silence as 0 and a message or collision as 1.

We are interested in the feasibility of deterministic broadcast in radio networks. If the nodes of the network do not have any labels {(or all have the same label)}, then deterministic broadcast is impossible even in the four-cycle. Indeed, the two neighbours of the source must behave identically, i.e., transmit in exactly the same rounds, and hence, due to collisions, the fourth node can never
hear a message.
On the other hand, if all nodes have distinct labels, then broadcast can be carried out, e.g., in a round-robin fashion, and hence $O(\log n)$-bit labels are sufficient for this task in $n$-node networks. It is easy to see that, by using a proper colouring of the square of the graph, $O(\log \Delta)$-bit labels, where $\Delta$ is the maximum degree, are enough to successfully broadcast. Hence, it is natural to ask if very short labels are sufficient for deterministic broadcast. In particular, is it possible to broadcast in every radio network
using labels of constant length? Below we formalize our question.

A {\em labeling scheme} for a network represented by a graph $G=(V,E)$ is any function $\cL$ from the set $V$ of nodes into the set $S$ of finite binary strings. The string $\cL(v)$ is called the \emph{label} of the node $v$.
Note that labels assigned by a labeling scheme are not necessarily distinct. The {\em length} of a labeling scheme $\cL$ is the maximum length of any label assigned by it.

Consider all graphs $G$, each labeled by some labeling scheme, with a distinguished source $s_G$. Initially, each node knows only its own label, and the source has a message. A {\em universal deterministic broadcast algorithm} works in synchronous rounds as follows. In each round, every node makes a decision if it should transmit or listen. This decision is based only on the current history of the node, which consists of the label of the node and the sequence of messages heard by the node so far. In particular, the decision does not depend on any knowledge of the graph $G$, including its size. However, the labeling scheme can use complete knowledge of the graph. Upon completion of the algorithm, all nodes should have the source message.
For simplicity, we assume that when a node transmits, it can transmit its entire history
(which may include the source message). However, in our algorithm, much smaller messages will suffice: they consist of either the source message or a constant-size ``stay" message.

We also consider a variant of the above problem called \emph{acknowledged broadcast}, which requires that the source node eventually knows that all nodes have received the source message. In our algorithm for this version of the problem, each transmitted message additionally contains a binary string of length $O(\log n)$, where $n$ is the size of the graph. One of the roles of this string is to implement a global clock. More specifically, in our algorithms, a node transmits only in response to receiving a message, and hence the current round number (which is the current local round number at the source node) can be maintained by including it in each transmitted message and incrementing it appropriately. Using our algorithm for acknowledged broadcast, we can ensure that there is a common round in which all nodes know that the broadcast of the source's message has been completed.

Using the above terminology, our central question can be formulated as follows:

\begin{quotation}
	Does there exist a universal deterministic (acknowledged) broadcast algorithm using labeling schemes of constant length for all radio networks? 
\end{quotation}

The above question can be seen in the framework of
algorithms using {\em informative labeling schemes}, or equivalently, algorithms with {\em advice} \cite{AKM01,CFIKP,DP,EFKR,FGIP,FIP1,FIP2,FKL,FPP,GPPR02,IKP,KKKP02,KKP05,SN}. 
When advice is given to nodes,  two variations are considered: either the binary string given to nodes is the same for all of them \cite{GMP} or different strings are given to different nodes
\cite{FKL,FPP}, as in the case
of the present paper. If strings may be different, they can be considered as labels assigned to nodes.
Several authors have studied the minimum amount of advice (i.e. label length) required to solve certain
network problems. The framework of advice or labeling schemes is useful for quantifying the amount of information used to solve a network problem, regardless of the type of information that is provided.

\subsection{Our results}
Our main contribution is a positive answer to our central question. For every radio network, we construct labeling schemes of constant length, and we design
universal deterministic broadcast and acknowledged broadcast algorithms using these schemes. For the broadcast task, our labeling schemes have length 2, while for acknowledged broadcast, our labeling schemes have length 3. {In the more difficult situation where the source node is not known at the time of labeling, our labeling scheme has length 3 (for both versions of broadcast).}

The importance of our result can be shown in the following scenario. Suppose that transmitting devices that form a radio network are already deployed, and only a central monitor knows the location and  the transmitting  range of each of them, thus knowing the topology of the resulting network. This could be applicable in an Internet of Things network in a business or industrial complex. 
One node of this network has to broadcast many consecutive messages to all other nodes. Then the monitor can assign very short labels to the devices, enabling multiple executions of the universal broadcast. The fact that labels can be very short may be crucial in situations when nodes of the network are weak and simple devices with very limited memory. Moreover, the fact that we can also do acknowledged broadcast in this situation permits the source to send the next message only after all nodes received the preceding one. 
Our work is also relevant in the context of Software-Defined Networks (SDNs) where the central controller assigns to each network device a role, i.e., a forwarding behaviour. Our solution gives an efficient implementation for broadcast that requires very few roles as well as simple forwarding rules.

\subsection{Related work}

Algorithmic problems in radio networks modeled as graphs were studied for such tasks as broadcast \cite{CGR,GPX1}, gossiping \cite{CGR,GPPR} and leader election
\cite{KP}. In some cases \cite{CGR,GPPR}, the topology of the network was unknown, in others \cite{GPX1}, nodes were assumed to have a labeled map of the network and could situate themselves in it.

For the broadcast task, most of the papers represented radio networks as
arbitrary (undirected or directed) graphs.  Models used in the
literature about algorithmic aspects of radio communication, starting
from \cite{CK}, differ mostly in the amount of information
about the network that is assumed available to nodes. However,
assumptions about this knowledge concern particular items of
information, such as the knowledge of the size of the network, its
diameter, maximum degree, or some neighbourhood around the nodes. In this paper, we adopt the approach of  limiting the total number of bits available to nodes,
regardless of their meaning.

Deterministic centralized broadcast assuming complete knowledge of
the network was considered in \cite{CW}, where a
polynomial-time algorithm constructing a $O(D \log ^2 n)$-time
broadcast scheme was given for all $n$-node networks of
radius~$D$. Subsequent improvements by many authors \cite{EK3,GM,GPX1}
were followed by the polynomial-time algorithm from \cite{KP5}
constructing a $O(D + \log ^2 n)$-time broadcast scheme, which is
optimal.  On the other hand, in \cite{ABLP}, the authors proved the
existence of a family of $n$-node networks of radius 2 for which any
broadcast algorithm requires time $\Omega (\log ^2 n)$. The ``minimal dominating sets" that appear in our work were used under the name ``minimal covering sets" in the context of deterministic centralized broadcast and gossiping assuming complete knowledge of the network \cite{GPX1,GPX2}.

One of the first papers to study deterministic distributed
broadcast in {\hyphenation{radio}radio} networks whose nodes have only
limited knowledge of the topology was~\cite{BGI}. The authors assumed
that nodes know only their own identifier and the identifiers of their neighbours.
Many authors \cite{BD,CGGPR,CGOR,CGR,CMS} studied deterministic
distributed broadcast in radio networks under the assumption that
nodes know only their own identifier (but not the identifiers of their neighbours).
In \cite{CGGPR}, the authors gave a broadcast algorithm working in time $O(n)$ for 
undirected $n$-node graphs, assuming that the nodes can transmit spontaneously before getting the source message.
For this model, a matching lower bound $\Omega(n)$ on deterministic broadcast time 
was proved in \cite{KP3}, even for the class of networks 
of constant diameter.
Increasingly faster broadcast algorithms working for arbitrary radio
networks were constructed, the currently fastest being the $O(n \log D \log\log D)$
algorithm from \cite{CzD}.  On the other hand, in \cite{CMS}, a lower
bound $\Omega(n \log D)$ on broadcast time was proved for $n$-node
networks of radius $D$.

Randomized broadcast algorithms in radio networks have also been studied \cite{BGI,KM}. For these algorithms, no topological knowledge
of the network and no labels of nodes were 
assumed. In \cite{BGI}, the authors showed a randomized broadcast
algorithm running in expected time $O(D \log n + \log ^2 n)$. 
In~\cite{KM}, it was shown that for any randomized broadcast
algorithm and parameters $D < n$, there exists an $n$-node network
of radius $D$ requiring expected time $\Omega(D \log(n/D))$ to execute
this algorithm.  It should be noted that the lower bound $\Omega(\log^2 n)$ 
from \cite{ABLP}, for some networks of radius 2, holds for
randomized algorithms as well. A randomized algorithm working in
expected time $O(D \log(n/D) + \log ^2 n)$, and thus matching the
above lower bounds, was presented in \cite{CR,KP4}.

Many papers \cite{AKM01,CFIKP,DP,EFKR,FGIP,FIP1,FIP2,FKL,FPP,GPPR02,KKKP02,KKP05,SN} have proposed algorithms to solve network tasks more efficiently by providing arbitrary information to nodes of the network or mobile agents circulating in it. These are known as algorithms using {\em informative labeling schemes} or algorithms with {\em advice}. Most relevant to this paper are those concerning radio networks. In \cite{IKP}, the authors considered the set of radio networks in which it is possible to perform broadcast in constant time when each node has complete knowledge of the network. They proved that $O(n)$ bits of advice are sufficient for performing broadcast in constant time in such networks and $\Omega(n)$ bits are necessary. Short labeling schemes have been found that can be used to perform topology recognition in radio networks modeled by trees \cite{GP1} and to perform size discovery in arbitrary radio networks with collision detection \cite{GP2}.

\section{Broadcast}\label{broadcasting}
In this section, we present a labeling scheme $\lambda$ that labels each node with a 2-bit string, and give a deterministic algorithm $\cB$ that solves broadcast on any graph $G$ that has been labeled using $\lambda$.

At a high level, broadcast is completed by having a set of ``informed" nodes, i.e., those that know the source message, that grows every two rounds. In odd-numbered rounds, we consider the set of ``frontier" nodes, i.e., uninformed nodes that are each adjacent to at least one informed node. From among the informed nodes, a minimal set of nodes that dominates the frontier nodes transmits the source message. Some of the frontier nodes will become newly-informed via these transmissions, while others will not, due to collisions. In even-numbered rounds, some of the newly-informed nodes will transmit a ``stay" message to inform certain nodes to stay in the dominating set for the next round. The first bit, $x_1$, of the label of a newly-informed node is used to determine whether or not it is added to the dominating set. The second bit, $x_2$, is used to determine whether or not it sends a ``stay" message. The formal description of our broadcast algorithm $\cB$ with source message $\mu$ is provided in Algorithm \ref{bcastpseudo}. We assume that there is a special ``stay" message that is distinct from the source message. Figure \ref{example} gives an example of an execution of $\cB$.

\begin{figure}[!ht]
	
	\begin{center}
		\includegraphics[scale=0.5]{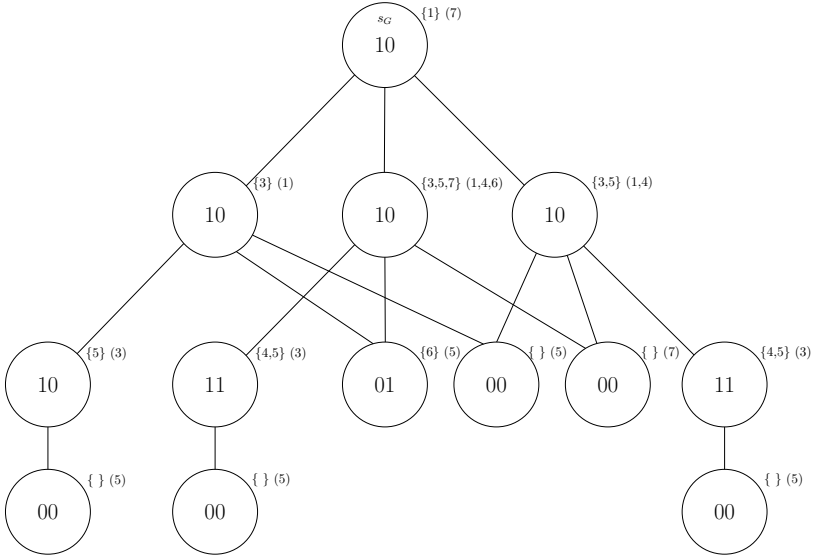}
	\end{center}
	
	\caption{\small Example of an execution of Algorithm $\cB$ on a graph labeled by $\lambda$. Each node contains its 2-bit label. To the upper right of each node: numbers in curly brackets are the round numbers in which the node transmits, numbers in parentheses are round numbers in which the node receives a message. Messages sent or received in odd rounds contain the source message $\mu$, and messages sent or received in even rounds contain ``stay".}
	\label{example}
\end{figure}

\begin{algorithm}[H]
	\small
	\caption{$\cB(\mu)$ executed at each node $v$}
	\label{bcastpseudo}
	\begin{algorithmic}[1]
		\Statex \% Each node has a variable \texttt{sourcemsg}. The source node has this variable initially set to $\mu$, all other nodes have it initially set to $null$.
		\For{each round $r$}
			\If{(never sent or received a message) {\bf and} (\texttt{sourcemsg} $\neq null$)}
				\Statex \algindent\algindent\% the source node transmits $\mu$ in first round
				\State transmit $\texttt{sourcemsg}$
			\ElsIf{(\texttt{sourcemsg} = $null$)}
				\Statex \algindent\algindent\% $v$ has not previously received $\mu$, listen for transmission
				\If{(message $m$ is received) {\bf and} $(m \neq \textrm{``stay"})$}
					\State  $\texttt{sourcemsg} \leftarrow m$
				\EndIf
			\Else 
				\Statex \algindent\algindent\% $v$ received $\mu$ before round $r$
				\If{$v$ first received \texttt{sourcemsg} in round $r-2$} \label{receivedtwoago}
					\If{$x_1 = 1$} 
						\State transmit $\texttt{sourcemsg}$\label{x1set}
					\EndIf
				\ElsIf{$v$ first received \texttt{sourcemsg} in round $r-1$} \label{receivedinlast}
					\If{$x_2 = 1$} 
						\State transmit $\textrm{``stay"}$ \label{x2set}
					\EndIf
				\ElsIf{$v$ transmitted $\texttt{sourcemsg}$ in round $r-2$ {\bf and} received $\textrm{``stay"}$ in round $r-1$} \label{receivedecho}
					\State transmit $\texttt{sourcemsg}$ \label{repeatsmsg}
				\EndIf
			\EndIf
		\EndFor
	\end{algorithmic}
\end{algorithm}

We now formally define the labeling scheme and prove the correctness of $\cB$. We rely heavily on five carefully chosen sequences of node sets. The following notation will be used in the construction of these sequences and throughout the remainder of this section. 

A set of nodes $X$ \emph{dominates} a set of nodes $Y$ if, for each node $y \in Y$, there is a node $x \in X$ that is adjacent to $y$.
For any set of nodes $X \subseteq V(G)$, denote by $\Gamma(X)$ the neighbourhood of $X$, i.e., $\Gamma(X) = \{v \in V(G)\ |\ \exists w \in X, \{v,w\} \in E(G) \}$.

\subsection{Sequence Constructions and Properties}

We construct five sequences of sets, indexed by $i \geq 1$. At a high level, $\inform{i}$ will be the nodes that are informed before round $2i-1$, $\uninf{i}$ will be the nodes that are not informed before round $2i-1$, $\frontier{i}$ will be the uninformed nodes that are adjacent to at least one informed node in round $2i-1$, $\new{i}$ will be the nodes that are newly-informed in round $2i-1$, and $\dom{i}$ will be the nodes that inform the nodes in $\new{i}$ in round $2i-1$. Recalling that $s_G$ denotes the source node of $G$, we initialize the construction by setting $\inform{1} = \{s_G\}, \uninf{1} = V(G) - \{s_G\}, \frontier{1} = \Gamma(s_G), \new{1} = \Gamma(s_G), \dom{1} = \{s_G\}$. Our construction proceeds in stages, where stage $i \geq 2$ is as follows:
\begin{enumerate}
	\item Define $\inform{i} = \inform{i-1} \cup \new{i-1}$.
	\item Define $\uninf{i} = \uninf{i-1} \setminus \new{i-1}$.
	\item Define $\frontier{i} = \uninf{i} \cap \Gamma(\inform{i})$.
	\item Define $\dom{i}$ to be a minimal subset of $\dom{i-1} \cup \new{i-1}$ that dominates all nodes in $\frontier{i}$.
	\item Define $\new{i}$ to be the subset of nodes in $\frontier{i}$ that are adjacent to exactly one node in $\dom{i}$.
\end{enumerate}

The construction ends when $\inform{i} = V(G)$. We now provide some useful facts about the sequences. The first two observations are direct consequences of the construction.

\begin{fact}\label{NewContainment}
	$\new{i} \subseteq \frontier{i} \subseteq \uninf{i}$ for all $i \geq 1$.
\end{fact}
\begin{fact}\label{unions}
	$\inform{i} = \inform{1} \cup \bigcup_{j=1}^{i-1} \new{j}$ and $\uninf{i} = \uninf{1} \setminus \bigcup_{j=1}^{i-1} \new{j}$.
\end{fact}

\begin{lemma}\label{NewDisjoint}
	For $i \neq i'$, we have $\new{i} \cap \new{i'} = \emptyset$.
\end{lemma}
\begin{proof}
	Without loss of generality, assume that $i > i'$. By Facts \ref{NewContainment} and \ref{unions}, it follows that $\new{i} \subseteq \uninf{i} = \uninf{1} \setminus \bigcup_{j=1}^{i-1} \new{j}$. In particular, $\new{i} \subseteq \uninf{1} \setminus \new{i'}$, so $\new{i} \cap \new{i'} = \emptyset$.
\end{proof}


The following result can be viewed as a guarantee of progress in each stage: if there are any remaining uninformed nodes at stage $i$, then at least one node will be newly informed in stage $i$.

\begin{lemma}\label{NewNonEmpty}
	For each $i \geq 1$, if $\inform{i} \neq V(G)$, then $\new{i} \neq \emptyset$.
\end{lemma}
\begin{proof}
	If $\inform{1} = \{s_G\} \neq V(G)$, then $\new{1} = \Gamma(s_G) \neq \emptyset$. So assume that $i \geq 2$. Since the graph is connected and $V(G)$ is the disjoint union of $\inform{i}$ and $\uninf{i}$, it follows that $\frontier{i} \neq \emptyset$.
	
	Consider any $v \in \dom{i}$. If each node $w \in \frontier{i}$ that is adjacent to $v$ is also adjacent to at least one other node in $\dom{i}$, then $\dom{i} \setminus \{v\}$ also dominates all nodes in $\frontier{i}$, contradicting the minimality of $\dom{i}$. Thus, there is at least one node $w \in \frontier{i}$ that is adjacent to $v$ and not adjacent to any other node in $\dom{i}$. Hence, by definition, $\new{i} \neq \emptyset$.
\end{proof}

The following result shows that the $\dom{i}$ is well-defined.
\begin{lemma}\label{DomExists}
	For all $i \geq 2$, there exists a subset of $\dom{i-1} \cup \new{i-1}$ that dominates all nodes in $\frontier{i}$.
\end{lemma}
\begin{proof}
	Consider any node $v \in \frontier{i}$ and suppose that $v$ does not have a neighbour in $\dom{i-1}$. By definition, $\dom{i-1}$ dominates all nodes in $\frontier{i-1}$, so it follows that $v \notin \frontier{i-1}$. By Fact \ref{NewContainment}, $v \in \frontier{i} \subseteq \uninf{i}$, and by construction, $\uninf{i} \subseteq \uninf{i-1}$, so $v \in \uninf{i-1}$. By the definition of $\frontier{i-1}$, it follows that $v \notin \Gamma(\inform{i-1})$. But, $v \in \frontier{i}$ implies that $v \in \Gamma(\inform{i})$. It follows that $v$ has a neighbour in $\inform{i} \setminus \inform{i-1} = \new{i-1}$. Therefore, we have shown that every node $v \in \frontier{i}$ has at least one neighbour in $\dom{i-1} \cup \new{i-1}$, which implies the desired result.
\end{proof}

Let $\laststage$ be the smallest value of $i$ such that $\inform{i} = V(G)$. We now give an upper bound on the value of $\ell$.

\begin{lemma}\label{BoundLastStage}
	$\ell \leq n$.
\end{lemma}
\begin{proof}
	The proof is by induction on $i$. By definition, $|\inform{1}|=1$, and, by Fact \ref{unions} and Lemma \ref{NewNonEmpty}, it follows that $|\inform{i}| \geq i$. Hence, $\inform{n} = V(G)$.
\end{proof}


It follows from Lemmas \ref{NewDisjoint} and \ref{BoundLastStage} that every node in $G \setminus \{s_G\}$ is contained in exactly one of the $\new{i}$ sets. We will later use this to ensure that all nodes in $G \setminus \{s_G\}$ are eventually informed.
\begin{corollary}\label{NewPartition}
	The sets $\new{1},\ldots,\new{\ell-1}$ form a partition of $G \setminus \{s_G\}$.
\end{corollary}
%
%
%
%

\subsection{The Labeling Scheme $\lambda$}\label{labeling}


Formally, our labeling scheme $\lambda(G)$ assigns a label $x_1x_2$ to each node in $G$ as follows:

\begin{itemize}
	\item
	For each node $v$, if there exists $i \geq 1$ such that $v \in \dom{i}$, then set $x_1 = 1$ at node $v$. Otherwise, set $x_1 = 0$ at node $v$.
	\item
	For each $i \geq 1$, for each node $v \in \dom{i+1} \cap \dom{i}$, arbitrarily pick one node $w \in \new{i}$ that is adjacent to $v$, and set $x_2 = 1$ at node $w$. At all other nodes, set $x_2 = 0$.
\end{itemize}

\subsection{Correctness of algorithm $\cB$}\label{analysis}
Our approach to showing that all nodes are eventually informed is to fully characterize which nodes transmit and which nodes are newly-informed in each round of the broadcast algorithm. Roughly speaking, we will show that, in an odd round $2i-1$, the nodes in $\dom{i}$ transmit and all nodes in $\new{i}$ receive the source message for the first time. Then, in round $2i$, a certain subset of $\new{i}$ transmits, which results in the nodes of $\dom{i+1}$ receiving ``stay". This will prompt the nodes of $\dom{i+1}$ to transmit in round $2i+1$, which informs all nodes in $\new{i+1}$. In this way, we will show that, for all $i \in \{1,\ldots,\ell-1\}$, all nodes in $\new{i}$ are informed in round $2i-1$. Since we have already shown that the sets $\new{1},\ldots,\new{\ell-1}$ partition $G \setminus \{s_G\}$, this will show that broadcast is completed.


\begin{lemma}\label{NewInformed}
	For each $t > 0$,
	\begin{enumerate}
		\item If $t=2i-1$, the following hold: 
		\begin{enumerate}
			\item
			Node $v$ transmits $\mu$ in round $t$ if and only if $v \in \dom{i}$. 
			\item
			Node $w$ receives $\mu$ for the first time in round $t$ if and only if $w \in \new{i}$.
		\end{enumerate}
		\item If $t=2i$, the following holds:
		\begin{enumerate}
			\item
			Node $v$ transmits ``stay" in round $t$ if and only if $v \in \new{i}$ and $v$'s label has $x_2 = 1$.
		\end{enumerate}
	\end{enumerate}
\end{lemma}
\begin{proof}
	The proof proceeds by induction on $t$. In the base case, $t=1$, we see that the source $s_G$ is the only node that transmits in round 1, it transmits $\mu$, and the set of nodes that receive $\mu$ for the first time in round 1 is $\Gamma(s_G)$. Since $\dom{1} = \{s_G\}$ and $\new{1} = \Gamma(s_G)$, this proves the base case.
	
	For a fixed $t \geq 2$, assume that the result holds for all rounds $t' < t$. The induction step has two cases:
	\begin{itemize}
		\item  $t = 2i$ for some $i \geq 1$.
		
		First, consider a node $v \in \new{i}$ such that $v$'s label has $x_2 = 1$. By the induction hypothesis, $v$ receives $\mu$ for the first time in round $2i-1$. By the definition of the broadcast algorithm, $v$ transmits ``stay" in round $2i$ at line \ref{x2set}.
		
		Conversely, suppose that $v$ transmits ``stay" in round $2i$. By the algorithm, $v$ must have transmitted ``stay" at line \ref{x2set}. From the code, it follows that $v$'s label has $x_2 = 1$ and $v$ received $\mu$ for the first time in round $2i-1$. By the induction hypothesis, $v \in \new{i}$. This completes the proof of 2(a).
		
		\item  $t = 2i-1$ for some $i \geq 2$.
		\begin{itemize}
			\item Proof of 1(a):
			
			First, suppose that $v \in \dom{i}$. By the definition of $\dom{i}$, we know that $\dom{i} \subseteq \dom{i-1} \cup \new{i-1}$. If $v \in \new{i-1}$, then, by the induction hypothesis, $v$ received $\mu$ for the first time in round $2i-3$. By the definition of the labeling scheme, we know that $v$'s label has $x_1 = 1$. Hence, by lines \ref{receivedtwoago}-\ref{x1set}, $v$ transmits $\mu$ in round $2i-1$. So, suppose $v \in \dom{i-1}$. By the induction hypothesis, $v$ transmitted $\mu$ in round $2i-3$. By the definition of the labeling scheme, there is exactly one node in $\new{i-1}$ that is adjacent to $v$ and is labeled with $x_2 = 1$. Therefore, exactly one neighbour of $v$ transmits ``stay" in round $2i-2$, so $v$ receives ``stay" in round $2i-2$. Hence, from lines \ref{receivedecho}-\ref{repeatsmsg}, $v$ transmits in round $2i-1$. 
			
			Conversely, suppose that $v$ transmits $\mu$ in round $2i-1$. There are two cases to consider, depending on whether $v$'s transmission of $\mu$ in round $2i-1$ occurred at line \ref{x1set} or \ref{repeatsmsg}. In the first case, by line \ref{x1set}, we know that $v$'s label has $x_1=1$, so, by the definition of labeling scheme, $v \in \dom{j}$ for some minimal $j$. Since $\dom{j} \subseteq \dom{j-1} \cup \new{j-1}$, the minimality of $j$ implies that $v \in \new{j-1}$. Further, by line \ref{receivedtwoago}, $v$ received $\mu$ for the first time in round $2i-3$. Hence, by the induction hypothesis, $v \in \new{i-1}$. By Lemma \ref{NewDisjoint}, it follows that $i=j$. Thus, $v \in \dom{i}$. Now, assume that $v$'s transmission occurred at line \ref{repeatsmsg}. By line \ref{receivedecho}, we know that $v$ received ``stay" in round $2i-2$. By the induction hypothesis, the nodes in $\new{i-1}$ with $x_2=1$ are the nodes that transmit ``stay" in round $2i-2$. It follows that $v$ is adjacent to exactly one node $w \in \new{i-1}$ whose label has $x_2=1$. By the definition of the labeling scheme, $w$ is adjacent to a node $v' \in \dom{i} \cap \dom{i-1}$. By the induction hypothesis, since $v' \in \dom{i-1}$, we know that $v'$ transmitted in round $2i-3$. By line \ref{receivedecho}, $v$ transmitted in round $2i-3$. Since $w \in \new{i-1}$, the induction hypothesis implies that $w$ received a message in round $2i-3$. Thus, $v=v' \in \dom{i}$.
			
			\item Proof of 1(b):
			
			First, suppose that $w$ receives $\mu$ for the first time in round $2i-1$. Since $w$ receives $\mu$ in round $2i-1$, it must be adjacent to exactly one node that transmits in round $2i-1$. By 1(a), we know that $\dom{i}$ is the set of nodes that transmit in round $2i-1$, which implies that $w$ is adjacent to exactly one node in $\dom{i}$. By the definition of $\new{i}$, it follows that $w \in \new{i}$.
			
			Conversely, suppose that $w \in \new{i}$. Then, by definition, $w$ is adjacent to exactly one node in $\dom{i}$. By 1(a), $\dom{i}$ is the set of nodes that transmit in round $2i-1$. It follows that $w$ receives message $\mu$ in round $2i-1$. If $w$ received $\mu$ for the first time in some round $t' < 2i-1$, then, by the induction hypothesis, $w$ is contained in some $\new{i'}$ where $i' < i$. This is impossible, by Lemma \ref{NewDisjoint}. Hence, $w$ received $\mu$ for the first time in round $2i-1$.
		\end{itemize}
	\end{itemize}
	
\end{proof}



We now prove that our algorithm ensures that all nodes in $G \setminus \{s_G\}$ are informed within $2n$ rounds.

\begin{theorem}\label{BcastCorrect}
	Consider any $n$-node unlabeled graph $G$ with a designated source node $s_G$ with source message $\mu$. By applying the 2-bit labeling scheme $\lambda$ and then executing algorithm $\cB$, all nodes in $G \setminus \{s_G\}$ are informed within $2n-3$ rounds.
\end{theorem}
\begin{proof}
	Consider an arbitrary node $w \in G \setminus \{s_G\}$. By Corollary \ref{NewPartition}, $w$ is contained in $\new{i}$ for exactly one $i \in \{1,\ldots,\ell-1\}$. By Lemma \ref{NewInformed}, $w$ receives $\mu$ for the first time in round $2i-1 \leq 2\ell-3$. By Lemma \ref{BoundLastStage}, we have that $2\ell - 3 \leq 2n - 3$, as desired.
\end{proof}

\section{Acknowledged Broadcast}\label{ackbroadcasting}
To solve acknowledged broadcast, we provide an algorithm $\cB_{ack}$ in which the source node $s_G$ receives an ``ack" message in some round $t$ after all nodes in $G \setminus \{s_G\}$ have received $\mu$. At a high level, $\cB_{ack}$ is obtained from $\cB$ by considering a particular node $z$ that receives $\mu$ last when $\cB$ is executed on $G$. An additional bit $x_3$ in each node's label is used to identify $z$. Once it receives $\mu$, node $z$ initiates the acknowledgement process by immediately transmitting an ``ack" message that contains the round number $k$ in which it first received $\mu$. The (unique) neighbour of $z$ that transmitted in round $k$ will receive this message, and it immediately transmits an ``ack" message that contains the round number $k'$ in which it first received $\mu$. This process continues until the source node receives an ``ack" message. The difficulty is that each node must know the round number in which it received $\mu$, and the round numbers in which it transmits. This is implemented as follows. The source node appends ``1" to its first transmitted message. Every other node determines the round number by recording the number that is appended to the first received message containing $\mu$, and appends the round number (appropriately increased) whenever it transmits. The formal description of our acknowledged broadcast algorithm $\cB_{ack}$ with source message $\mu$ is provided in Algorithm \ref{bcastackpseudo}. We assume that there is a special ``ack" message that is distinct from the source message and the ``stay" message.

\subsection{The Labeling Scheme $\lambda_{ack}$}\label{labelingack}

The labeling scheme is identical to $\lambda$ except that one node $z$ will have a new label. This can be represented using an additional bit, $x_3$, which is 1 for $z$ and 0 for all other nodes. {The node $z$ is chosen as follows: label $G$ using labeling scheme $\lambda$ and execute $\cB$ on the resulting labeled graph, then determine the first round $r$ after which there are no uninformed nodes, and choose $z$ to be a node that receives $\mu$ in round $r$. If there is more than one such node, choose $z$ arbitrarily from among them. 
	
We note that the labeling scheme $\lambda_{ack}$ will never assign certain labels to any node, which means we may safely use these labels in later schemes that are built on top of $\lambda_{ack}$. At a high level, this is because $z$ is the only node with bit $x_3$ set to 1, and, as there are no remaining uninformed nodes after $z$ receives $\mu$, our labeling scheme will set $z$'s bits $x_1$ and $x_2$ to 0 to indicate that $z$ should not transmit after receiving $\mu$.

\begin{fact}\label{labelsUnused}
	For any graph $G$, when the labeling scheme $\lambda_{ack}$ is applied to $G$, no node is labeled with 101 or 111 or 011.
\end{fact}
\begin{proof}
	By definition, node $z$ is the only node that has bit $x_3 = 1$, so it is sufficient to prove that node $z$ has bit $x_1 = x_2 = 0$. By the definition of $\lambda$ in Section \ref{labeling}, it is sufficient to prove that there is no value of $i \geq 1$ such that $z \in \dom{i}$. To obtain a contradiction, assume that there exists an $i \geq 1$ such that $z \in \dom{i}$, and let $j$ be the smallest such $i$. Then, by the definition of $\dom{j}$, the fact that $z \in \dom{j}$ implies that $z \in \new{j-1}$. By the choice of $z$ by $\lambda_{ack}$ and Lemma \ref{NewInformed}, node $z$ receives $\mu$ for the first time in round $2(j-1)-1$, and there are no remaining uninformed nodes after this round. In particular, by Lemma \ref{NewInformed}, this means that $\new{i} = \emptyset$ for all $i \geq j$, and so $\inform{j} = V(G)$ by Lemma \ref{NewNonEmpty}. This implies that $\uninf{j} = \emptyset$, so $\frontier{j} = \emptyset$, which means $\dom{j} = \emptyset$, which contradicts the fact that $z \in \dom{j}$. 
\end{proof}
}

\subsection{Correctness of algorithm $\cB_{ack}$}\label{analysisack}

To prove the correctness of $\cB_{ack}$, we first observe that all transmissions of ``ack" messages occur after all transmissions of $\mu$ and ``stay" messages, i.e., the broadcast and the acknowledgement process do not interfere with one another. The first two observations follow from Lemma \ref{NewInformed} and the fact that $\new{i} = \dom{i} = \emptyset$ for all $i \geq \ell$. 

\begin{obs}\label{lastround}
	The last round in which a node receives $\mu$ for the first time is $2\ell-3$.
\end{obs}

\begin{obs}\label{onlyacks}
	No transmissions of $\mu$ nor ``stay" occur after round $2\ell-3$.
\end{obs}

The next observation follows from Observation \ref{lastround} and the definitions of algorithms $\lambda_{ack}$ and $\cB_{ack}$.

\begin{obs}\label{ztransmits}
	The first transmission of ``ack" occurs in round $2\ell-2$, and is transmitted by the unique node $z$ whose label has $x_3 = 1$.
\end{obs}

\begin{algorithm}[H]
	\caption{$\cB_{ack}(\mu)$ executed at each node $v$}
	\label{bcastackpseudo}
	
	\begin{algorithmic}[1]
		
		\Statex\% Each node has a variable \texttt{sourcemsg}. The source node has this variable initially set to $\mu$, all other nodes have it initially set to $null$. Each node maintains a variable \texttt{informedRound} that keeps track of the first round in which it received $\mu$. Each non-source node maintains a variable \texttt{transmitRounds} that keeps track of the set of rounds in which it transmitted $\mu$.
		\State $\texttt{informedRound} \leftarrow null$
		\State $\texttt{transmitRounds} \leftarrow null$
		\For{each round $r$}
			\If{(never sent or received a message) {\bf and} (\texttt{sourcemsg} $\neq null$)}
		\Statex \algindent\algindent\% source node transmits $\mu$ in first round
		\State  transmit $(\texttt{sourcemsg},1)$
		\ElsIf{(\texttt{sourcemsg} = null)}
		\Statex \algindent\algindent\% has not previously received $\mu$, listen for transmission
		\If{(message $(m,k)$ is received) {\bf and} $(m \neq \textrm{``stay"})$}
		\State  $\texttt{sourcemsg} \leftarrow m$
		\State $\texttt{informedRound} \leftarrow k$ \label{ack-setinformed}
		\EndIf
		\Else
		\Statex \algindent\algindent\% the node received $\mu$ before round $r$
		\If{$v$ first received \texttt{sourcemsg} in round $r-2$} \label{ack-receivedtwoago}
		\If{$x_1 = 1$}\label{ack-x1set}
		\State transmit $(\texttt{sourcemsg},\texttt{informedRound}+2)$ \label{ack-transmitmsg}
		\State insert $\texttt{informedRound}+2$ into $\texttt{transmitRounds}$
		\EndIf
		\ElsIf{$v$ first received \texttt{sourcemsg} in round $r-1$}\label{ack-receivedinlast}
		\If{$x_3 = 1$}\label{ack-x3set}
		\Statex \algindent\algindent\algindent\algindent\% start acknowledgement process
		\State transmit $(\textrm{``ack"},\texttt{informedRound})$ \label{ack-startack}
		\ElsIf{$x_2 = 1$}\label{ack-x2set}
		\State transmit $(\textrm{``stay"},\texttt{informedRound}+1)$ \label{ack-transmitsecho}
		\EndIf
		\ElsIf{$v$ received $(\textrm{``stay"},k)$ in round $r-1$}\label{ack-receivedecho}
		\If{$v$ transmitted $\texttt{sourcemsg}$ in round $r-2$} \label{ack-transmittedtwoago}
		\State transmit $(\texttt{sourcemsg},k+1)$ \label{ack-repeatsmsg}
		\State insert $k+1$ into $\texttt{transmitRounds}$
		\EndIf
		\ElsIf{$v$ received $(\textrm{``ack"},k)$ in round $r-1$}\label{ack-receivedack}
		\If{$k$ is contained in \texttt{transmitRounds}} \label{ack-didtransmit}
		\State transmit $(\textrm{``ack"},\texttt{informedRound})$ \label{ack-repeatack}
		\EndIf
		\EndIf
		\EndIf
		\EndFor
	\end{algorithmic}
\end{algorithm}

%

\newpage
We now prove that the correct round number is appended to each message containing $\mu$, which is necessary for the correctness of the acknowledgement process.

\begin{lemma}\label{globalclock}
	The messages $(\mu,t)$ and $(``stay",t)$ are transmitted only in round $t$.
\end{lemma}
\begin{proof}
	The proof is by induction on $t$. For the base case, $t=0$, we see that the source node sends $(\mu,0)$ in its first transmission, and no other nodes transmit before receiving $\mu$ for the first time. As induction hypothesis, assume that, for all $0 \leq t' < t$, a message $(\mu,t')$ or $(\textrm{``stay"},t')$ is only transmitted in round $t'$.
	
	First, suppose that a node $v$ transmits a message $(\textrm{``stay"},t)$. This occurs at line \ref{ack-transmitsecho}, which, by line \ref{ack-receivedinlast}, implies that $v$ received $\mu$ for the first time in round $t-1$ via some message $(\mu,t')$. By the induction hypothesis, $t' = t-1$. Therefore, $v$ sets \texttt{informedRound} equal to $t-1$ at line \ref{ack-setinformed}. So, when $v$ transmits $(\textrm{``stay"},\texttt{informedRound}+1)$, it follows that $\texttt{informedRound}+1= t$, as desired.
	
	Next, suppose that a node $v$ transmits a message $(\mu,t)$. If this transmission occurs at line \ref{ack-transmitmsg}, then, by line \ref{ack-receivedtwoago}, we know that $v$ received $\mu$ for the first time in round $t-2$ via some message $(\mu,t')$. By the induction hypothesis, $t' = t-2$. Therefore, $v$ sets \texttt{informedRound} equal to $t-2$ at line \ref{ack-setinformed}. So, when $v$ transmits $(\textrm{``stay"},\texttt{informedRound}+2)$, it follows that $\texttt{informedRound}+2 = t$, as desired. The other possibility is that the transmission by $v$ occurs at line \ref{ack-repeatsmsg}, which, by line \ref{ack-receivedecho}, implies that $v$ received a $(\textrm{``stay"},t')$ message in round $t-1$. By the induction hypothesis, $t'=t-1$.  So, when $v$ transmits $(\mu,t'+1)$, it follows that $t'+1 = t$, as desired.
\end{proof}

From Lemma \ref{globalclock}, it follows that if a node $v \neq s_G$ first receives $\mu$ in round $t$, then the \texttt{informedRound} variable at node $v$ is equal to $t$ in all subsequent rounds. Similarly, if a node $v \neq s_G$ transmits a message containing $\mu$ in round $t$, then the \texttt{transmitRounds} variable at node $v$ contains $t$ in all subsequent rounds. 

%

We complete the proof of correctness of $\cB_{ack}$ by showing that the source node will eventually receive an ``ack" message. First, we show that at most one node transmits ``ack" in any round, which implies that no collisions occur during the acknowledgement procedure.

\begin{lemma}\label{atmostone}
	After round $2\ell-3$, at most one node $v$ transmits in each round.
\end{lemma}
\begin{proof}
	The proof is by induction on the round number $t$. For the base case, Observations \ref{lastround}-\ref{ztransmits} imply that the unique node $z$ with $x_3 = 1$ in its label transmits $(\textrm{``ack"},2\ell-3)$ in round $2\ell-2$. As induction hypothesis, assume that at most one node transmits in round $t \geq 2\ell-2$. If no node transmits in round $t$, then, from Observation \ref{onlyacks} and the code, no node transmits in round $t+1$. Otherwise, suppose that exactly one node $v$ transmits in round $t$. By Observation \ref{onlyacks}, this message is of the form $(\textrm{``ack"},k)$. At most one neighbour $w$ of $v$ contains $k$ in its \texttt{transmitRounds} variable since $v$ received $\mu$ in round $k$. From Observation \ref{onlyacks} and the code, no other node transmits in round $t+1$.
\end{proof}

We now show that the ``ack" message propagates through a sequence of nodes, where each node is contained in some $\dom{i}$. Further, the indices of the corresponding $\dom{i}$ sets form a decreasing sequence, which implies that $\{s_G\} = \dom{1}$ will eventually receive an ``ack" message.

\begin{lemma}\label{acksequence}
	For each $i \in \{0,\ldots,\ell-2\}$, in round $2\ell-2+i$, some node in $\dom{j}$ with $j \leq \ell-i-1$ receives $(\textrm{``ack"},2j-1)$.
\end{lemma}
\begin{proof}
	The proof is by induction on $i$. For the base case $i=0$, Observations \ref{lastround} and \ref{ztransmits} imply that $z$ transmits an $(\textrm{``ack"},2\ell-3)$ message in round $2\ell-2$. By Lemma \ref{atmostone}, no other node transmits in round $2\ell-2$, so all of $z$'s neighbours receive the transmitted ``ack" message. Since $z$ received $\mu$ in round $2\ell-3$, it follows that a neighbour $z'$ of $z$ transmitted $\mu$ in round $2\ell-3$. By Lemma \ref{NewInformed}, $z' \in \dom{\ell-1}$. Therefore, the statement is satisfied with $j=\ell-1$, which completes the base case.
	
	As induction hypothesis, assume that, for some $i \in \{0,\ldots,\ell-1\}$, in round $2\ell-2+i$, some node $w \in \dom{j}$ with $j \leq \ell-i-1$ receives $(\textrm{``ack"},2j-1)$. Since $w \in \dom{j}$, Lemma \ref{NewInformed} implies that $w$ transmitted $\mu$ in round $2j-1$. Therefore, its \texttt{transmitRounds} variable contains $2j-1$. By lines \ref{ack-receivedack}-\ref{ack-repeatack} of $\cB_{ack}$, it follows that $w$ transmits $(\textrm{``ack"},\texttt{informedRound})$ in round $2\ell-1+i$. We note that the value of \texttt{informedRound} at $w$ must be less than $2j-1$, since $w$ must have received $\mu$ for the first time before it transmitted $\mu$ in round $2j-1$. From Lemma \ref{NewInformed}, we conclude that $\texttt{informedRound} = 2j'-1$ for some $j' < j$. So we have shown that $w$ transmits $(\textrm{``ack"},2j'-1)$ for some $j' < j$ in round $2\ell-1+i$. By Lemma \ref{atmostone}, no other node transmits in round $2\ell-1+i$, so all of $w$'s neighbours receive the transmitted ``ack" message in round $2\ell-1+i$. Since $w$ received $\mu$ in round $2j'-1$, it follows that a neighbour $w'$ of $w$ transmitted $\mu$ in round $2j'-1$. By Lemma \ref{NewInformed}, $w' \in \dom{j'}$. To summarize, we have shown that in round $2\ell-2+(i+1)$, some node $w' \in \dom{j'}$ with $j' \leq j-1 \leq \ell-i-2 = \ell-(i+1)-1$ receives $(\textrm{``ack"},2j'-1)$, which completes the induction.
\end{proof}

\begin{corollary}\label{sourceack}
	There exists a round $t \in \{2\ell-2,\ldots,3\ell-4\}$ in which the source node receives an ``ack" message.
\end{corollary}

The correctness of $\cB_{ack}$ follows directly from Corollary \ref{sourceack}, which gives us the main result of this section.

\begin{theorem}\label{BcastAckCorrect}
	Consider any $n$-node unlabeled graph $G$ with a designated source node $s_G$ with source message $\mu$. By applying the 3-bit labeling scheme $\lambda_{ack}$ and then executing algorithm $\cB_{ack}$, all nodes in $G \setminus \{s_G\}$ are informed by round $t \leq 2n-3$, and $s_G$ receives an ``ack" message by round $t' \in \{t+1,\ldots,t+n-2\}$.
\end{theorem}

Finally, we note that $\cB$ and $\cB_{ack}$ can be used to ensure that there is a common round in which all nodes know that the broadcast of the source's message $\mu$ has been completed. First, run $\cB_{ack}$, and have the source node record the round number $m$ in which it first receives an ``ack" message. Then, the source executes $\cB$ with message $m$. All nodes will receive the value of $m$ before round $2m$. So, in round $2m$, all nodes know that the original broadcast of $\mu$ has been completed.

{
\section{Broadcast from an Arbitrary Source}

In this section, we consider the more difficult scenario in which the source node is not designated in $G$ when the labeling scheme is applied. We provide a labeling scheme of length 3 and a universal deterministic algorithm $\mathcal{B}_{arb}$ that solves (acknowledged) broadcast regardless of which node initially knows the source message. 

\subsection{The Labeling Scheme $\lambda_{arb}$}
	Choose an arbitrary node $r$ and label this node using the string 111. Apply the labeling scheme $\lambda_{ack}$ to the remaining nodes in the network, but use $r$ as the source node (as there is no designated source $s_G$). By Fact \ref{labelsUnused}, note that $\lambda_{ack}$ does not assign the label 111 to any node, so the node $r$ is a unique node in the network that our algorithm can use to play a special role in coordinating the broadcast, regardless of which node is the actual source $s_G$. Let $z$ be the node labeled 001 by $\lambda_{ack}$, i.e., the node that initiates the acknowledgement process in an execution of $\mathcal{B}_{ack}$.

\subsection{Algorithm $\mathcal{B}_{arb}$}
\begin{enumerate}
	\item Perform an acknowledged broadcast using $\mathcal{B}_{ack}$ with node $r$ as source and with message ``initialize". Each node $v$ stores in a variable $t_v$ the timestamp value contained in the first ``initialize" message it received. In particular, node $r$ sets $t_r$ to 0. When starting the acknowledgement process, node $z$ appends to the ``ack" message the timestamp value $T = t_z$. This step of the algorithm ends when $r$ receives the ``ack" message, at which point it knows the value of $T$ and it knows that all nodes have received ``initialize".
	\item Perform an acknowledged broadcast using a modified version of $\mathcal{B}_{ack}$ with node $r$ as the source and with message (``ready",$T$). The modification to the algorithm is that the node $z$ does not initiate the acknowledgement process. Instead, when the source node $s_G$ receives the ``ready" message, it waits $T$ rounds, then initiates the acknowledgement process (as described in $\mathcal{B}_{ack}$), but with the source message $\mu$ appended to the ``ack" message. (Waiting $T$ rounds ensures that this acknowledgement process started by $s_G$ does not begin until the ``ready" broadcast has completed.) This step of the algorithm ends when $r$ receives the ``ack" message, at which point it knows the source message $\mu$. Further, all nodes know the value of $T$.
	\item Perform a broadcast using $\mathcal{B}$ with node $r$ as source and with message $\mu$. At the end of this broadcast, all nodes know the source message $\mu$. If each node $v$ waits $T - t_v$ rounds after receiving $\mu$ in this step, then the algorithm solves acknowledged broadcast, as all nodes can be sure that all nodes have received $\mu$.
\end{enumerate}
%

}

\section{Conclusion}

We presented a universal deterministic broadcast algorithm using labeling schemes of constant length that works for arbitrary radio networks. Our schemes are of length 2, and we showed how to solved acknowledged broadcast with schemes of length 3 (but only 5 different labels). {In the case where the source node is not designated when the labeling scheme is applied, our scheme also has length 3, but uses 6 different labels. It would be interesting to determine if schemes using fewer than 4 different labels are sufficient for broadcast. We do not have any impossibility results beyond the trivial 1-bit lower bound (2 different labels), and we are intrigued by the possibility that there exists a scheme of length 1 for broadcast.} A positive answer can be obtained for broadcast in graphs where each node's distance to the source is at most 2: use $\lambda$ and $\cB$ from Section \ref{broadcasting}, but use only the bit $x_2$, and modify the definitions of $\frontier{i}$ and $\dom{i}$ by changing instances of $\dom{i-1} \cup \new{i-1}$ to $\dom{i-1}$. We can also show that it is possible to perform broadcast in series-parallel graphs {and grid graphs} using single-bit labels. In both cases, using the same technique from Section \ref{ackbroadcasting}, acknowledged broadcast is possible using 3 labels. It would also be interesting to determine whether or not acknowledged broadcast is possible in all graphs using 1-bit labels, and, if not, if it is possible using 2-bit labels. {Another open question is whether acknowledged broadcast can be performed with only constant-length 
	messages, instead of $O(\log n)$ bits. The above open questions can also be asked for the case where the source node is not designated when the labeling scheme is applied.}

In this paper, we focused on the feasibility of radio broadcast with short labels, and we did not try to optimize the time complexity. Our algorithm works in time $O(n)$. This yields the following open problem. What is the fastest universal deterministic broadcast algorithm using labeling schemes of constant length?